\documentclass[conference]{IEEEtran}
\makeatletter
\def\ps@headings{%
\def\@oddhead{\mbox{}\scriptsize\rightmark \hfil \thepage}%
\def\@evenhead{\scriptsize\thepage \hfil \leftmark\mbox{}}%
\def\@oddfoot{}%
\def\@evenfoot{}}
\makeatother
\usepackage{placeins}
\usepackage{amsthm}
\usepackage{float}
\usepackage{graphicx}
\usepackage{subcaption}
\usepackage{psfrag}
\usepackage[ruled,vlined,linesnumbered,commentsnumbered]{algorithm2e}
\usepackage{array}
\usepackage{amsmath,amssymb}
\usepackage{wasysym}
\usepackage{paralist}
\usepackage{color}
\usepackage{epstopdf}
\usepackage{float}
\usepackage[font={small}]{caption}

\newtheorem{theorem}{Theorem}

\newenvironment{definition}[1][Definition:]{\begin{trivlist}
\item[\hskip \labelsep {\bfseries #1}]}{\end{trivlist}}

\begin{document}

\title{On the Entity Hardening Problem in Multi-layered Interdependent Networks}

\author{\IEEEauthorblockN{Joydeep Banerjee, Arun Das, Chenyang Zhou, Anisha Mazumder and Arunabha Sen}
\IEEEauthorblockA{ Computer Science and Engineering Program\\
\small School of Computing, Informatics and Decision System Engineering\\
\small Arizona State University\\
\small Tempe, Arizona 85287\\
\small Email: \{joydeep.banerjee, arun.das, czhou24, anisha.mazumder, asen\}@asu.edu}
}

\maketitle

\begin{abstract}
The power grid and the communication network are highly interdependent on each other for their well being. In recent times the research community has shown significant interest in modeling such interdependent networks and studying the impact of failures on these networks. Although a number of models have been proposed, many of them are simplistic in nature and fail to capture the complex interdependencies that exist between the entities of these networks. To overcome the limitations, recently an \emph{Implicative Interdependency Model} that utilizes Boolean Logic, was proposed and a number of problems were studied. In this paper we study the ``\emph{entity hardening}'' problem, where by ``\emph{entity hardening}'' we imply the ability of the network operator to ensure that an adversary (be it Nature or human) cannot take a network entity from \emph{operative} to \emph{inoperative} state. Given that the network operator with a limited budget can only harden $k$ entities, the goal of the entity hardening problem is to identify the set of $k$ entities whose hardening will ensure maximum benefit for the operator, i.e. maximally reduce the ability of the adversary to degrade the network. We show that the problem is solvable in polynomial time for some cases, whereas for others it is NP-complete. We provide the optimal solution using ILP, and propose a heuristic approach to solve the problem. We evaluate the efficacy of our heuristic using power and communication network data of Maricopa County, Arizona. The experiments show that our heuristic almost always produces near optimal results.

\end{abstract}

\section{Introduction}
The critical infrastructures of a nation form a complex symbiotic ecosystem where individual infrastructures are heavily interdependent on each other for being fully functional. Two such critical systems that rely heavily on each other for their well being are the power and communication network infrastructures. For instance, power grid entities such as SCADA systems, that are used to remotely operate power generation units, receive their control commands over the communication network infrastructure, while communication network entities such as routers and base stations are inoperable without electric power. Thus, failure introduced in the system either by Nature (hurricanes), or man (terrorist attacks), can trigger further failures in the system due to interdependencies between the entities of the two infrastructures.

Although a number of models have been proposed for modeling and analysis of interdependent multi-layered networks \cite{Bul10, Gao11, Sha11, Ros08, Zha05, Par13, Ngu13, Zus11}, many of these models are simplistic in nature and fail to capture the complex interdependencies that exists between the entities of these networks. As noted in \cite{sen2014identification}, these models fail to model complex interdependencies that may exist between network entities, such as when entity $a_i$ is operational, if entities (i) $b_j$ {\em and} $b_k$ {\em and} $b_l$ are operational, {\em or} (ii) $b_m$ {\em and} $b_n$ are operational, {\em or} (iii) $b_p$ is operational. Graph based interdependency models proposed in the literature such as \cite{Sha11, Ros08, Zha05, Cas13, Par13, Ngu13} including \cite{Bul10, Gao11} cannot capture such complex interdependency involving both conjunctive and disjunctive terms between entities of multi-layer networks. To overcome these limitations, an \emph{Implicative Interdependency Model} that utilizes Boolean Logic, was recently proposed in \cite{sen2014identification}, and a number of problems including \emph{computation of $\cal{K}$ most vulnerable nodes} \cite{sen2014identification}, \emph{root cause of failure analysis} \cite{das2014root}, and \emph{progressive recovery from failures} \cite{mazumder2014progressive}, were studied using this model. 

In this paper we study the ``\emph{entity hardening}'' problem in the interdependent power-communication network using the \emph{Implicative Interdependency Model} (IIM). By ``\emph{entity hardening}'', we imply the ability of the network operator to ensure that an adversary (be it Nature or human), cannot take a network entity from an \emph{operative} (operational) to an \emph{inoperative} (failed) state. We assume that the adversary is clever and is capable of identifying the most vulnerable entities in the network that causes maximum damage to the interdependent system. However, the adversary does not have an unlimited budget and has the resources to destroy at most $\cal{K}$ entities of the interdependent network. The network operator is also aware of adversary's target entities for destruction. Since we assume that once an entity is ``\emph{hardened}'' by the network operator it cannot be destroyed by the adversary, if all $\cal{K}$ targets of the adversary are hardened by the network operator, then the adversary cannot induce any failure in the network. However, if due to resource limitations the network operator is able to strengthen only $k$ entities, where $k < \cal{K}$, these $k$ entities have to be carefully chosen. The goal of the entity hardening problem is to identify the set of $k$ entities whose hardening will ensure maximum benefit for the operator, i.e. maximally reduce the ability of the adversary to degrade the network. 

We classify the entity hardening problem into four different cases depending on the nature of the interdependency relationships. We show that the first case can be solved in polynomial time, and all other cases are shown to be NP-complete. We provide an \emph{inapproximability} result for the second case, an \emph{approximation algorithm} for the third case, and a heuristic for the fourth (general) case. We evaluate the efficacy of our heuristic using power and communication network data of Maricopa County, Arizona. The experiments show that our heuristic almost always produces near optimal results.

The paper is organized as follows, the IIM model is presented in Section \ref{IIM_Section}, in Sections \ref{ProbForm} and \ref{CompAna} we formally state the entity hardening problem and analyze its computational complexity, Section \ref{Solutions} outlines the optimal and heuristic solutions to the problem, Section \ref{ExpRes} shows the experimental results, and finally Section \ref{Conclusion} concludes this paper.

\section{Interdependency Model}
\label{IIM_Section}
We now present an overview of the underlying IIM interdependency model \cite{sen2014identification}. IIM uses Boolean Logic to model the interdependencies between network entities, these interdependent relationships are termed as {\em Implicative Interdependency Relations} (IDRs). We represent this interdependent network setting as $\mathcal{I}(A,B,\mathcal{F}(A,B))$, where sets $A$ and $B$ are the power and communication network entities respectively, and $\mathcal{F}(A,B)$ is the set of dependency relations, or IDRs. Table \ref{tbl:example1idr} represents a sample interdependent network $\mathcal{I}(A,B,\mathcal{F}(A,B))$, where $A=\{a_1,a_2,a_3,a_4\}$, $B=\{b_1,b_2,b_3\}$ and $\mathcal{F}(A,B)$ is the set of IDRs (dependency relations) between the entities of $A$ and $B$. In this example, the IDR $b_1 \leftarrow a_1a_3 + a_2$ implies that entity $b_1$ is operational when both the entities $a_1$ \emph{and} $a_3$ are operational, \emph{or} entity $a_2$ is operational. The conjunction of entities, such as $a_1a_3$, is also referred to as a \emph{minterm}.

\begin{table}[H]
\begin{center}
\begin{tabular}{|l||l|}  \hline
{\bf Power Network} & {\bf Comm. Network} \\ \hline
$a_1\leftarrow b_1b_2$ & $b_1 \leftarrow a_1a_3 + a_2$ \\ \hline
$a_2 \leftarrow  b_1 + b_2$ & $b_2 \leftarrow a_1a_2a_3$ \\ \hline
$a_3 \leftarrow b_1 + b_2 + b_3$ & $b_3 \leftarrow a_1 + a_2 + a_3$ \\ \hline
$a_4 \leftarrow b_1 + b_3$ & $--$ \\ \hline
\end{tabular}
\caption{Implicative Interdependency Relations of a sample network}
\protect\label{tbl:example1idr}
\end{center}
\vspace{-25pt}
\end{table}

Given a set of inoperable (failed) entities, a time stepped failure cascade can be derived from the dependency relationships outlined in the IDR set. For example, for the interdependent network outlined in Table \ref{tbl:example1idr},  Table \ref{tbl:example1cascade} shows the failure propagation when entities $\{a_2, b_3\}$ fail at the initial time step ($t=0$). It may be noted that the model assumes that dependent entities fail immediately in the next time step, for example, when $\{a_2, b_3\}$ fail at $t=0$, $b_2$ fails at $t=1$ as $b_2$ is dependent on $a_2$ for its survival. The system reaches a \emph{steady state} when the failure propagation process stops. In this example, when $\{a_2, b_3\}$ fail at $t=0$, the steady state is reached at time step $t=4$.

\begin{table}[h]
\begin{center}
\begin{tabular}{|c|c|c|c|c|c|c|c|}  \hline
\multicolumn{1}{|c|}{Entities} & \multicolumn{7}{c|}{Time Steps ($t$)}\\
\cline{2-8} & $0$ & $1$ & $2$ & $3$ & $4$ & $5$ & $6$ \\\hline \hline
$a_1$ & $0$ & $0$ & $1$ & $1$ & $1$ & $1$ & $1$ \\ \hline
$a_2$ & $1$ & $1$ & $1$ & $1$ & $1$ & $1$ & $1$ \\ \hline
$a_3$ & $0$ & $0$ & $0$ & $0$ & $1$ & $1$ & $1$ \\ \hline
$a_4$ & $0$ & $0$ & $0$ & $0$ & $1$ & $1$ & $1$ \\ \hline
$b_1$ & $0$ & $0$ & $0$ & $1$ & $1$ & $1$ & $1$ \\ \hline
$b_2$ & $0$ & $1$ & $1$ & $1$ & $1$ & $1$ & $1$ \\ \hline
$b_3$ & $1$ & $1$ & $1$ & $1$ & $1$ & $1$ & $1$ \\ \hline
\end{tabular}
\caption{Failure cascade propagation when entities $\{a_2, b_3\}$ fail at time step $t=0$. A value of $1$ denotes entity failure, and $0$ otherwise}
\protect\label{tbl:example1cascade}
\end{center}
\vspace{-20pt}
\end{table}

A primary consideration for using this model is the accurate formulation of the IDRs that is representative of the underlying physical power and communication network infrastructures. This can either be done by careful analysis as done in \cite{Zus11}, or by consultation with experts of these infrastructures. We utilize IIM to model the interdependency between the two networks and analyze the entity hardening problem in this setting.

\section{Problem Formulation}
\label{ProbForm}

Before we make a formal statement of the entity hardening problem in the IIM setting, we explain it with the help of an example. Consider an interdependent system as outlined in the IDR set shown in Table \ref{tbl:example1idr}. It may be easily checked that when the adversary budget is $\cal{K}$$=2$, the most vulnerable entities of this system are $\{a_2, b_3\}$. If the network operator doesn't harden any one of the entities $a_2$ or $b_3$, then in this example all the network entities eventually fail, as seen from the fault propagation in Table \ref{tbl:example1cascade}. When the network operator chooses to harden both $a_2$ and $b_3$ then none of the entities in the network fail if the adversary restricts the attack only to the two most vulnerable entities of the network, which in this example happens to be $\{a_2, b_3\}$. If the network operator has resources to harden only one entity and the operator chooses to harden $a_2$, the destruction of $b_3$ by the adversary will eventually lead to the failure of no other entities of the network, as shown in Table \ref{tbl:example1modcascade}(\subref{tbl:example1modcascade_a}). If on the other hand, the network operator chooses to harden $b_3$, destruction by the adversary of $a_2$ will eventually lead to the failure of the entities $\{a_2,b_2,a_1,b_1\}$ as shown in Table \ref{tbl:example1modcascade}(\subref{tbl:example1modcascade_b}). Clearly in this scenario the operator should harden $a_2$ instead of $b_3$.



\begin{definition}
\emph{Kill Set of a set of Entities($S$):} The kill set of a set of entities $S$, is the set of all entities that will eventually fail due to failure of $S$ and the interdependencies between the entities of the network as given by the set of IDR's. The kill set of a set of entities $S$ is denoted by $KillSet(S)$.
\end{definition}

It may be noted that the search for $k$ entities to be hardened is restricted to the $KillSet(S)$, where $S$ is the set of $\cal{K}$ most vulnerable entities in the network, because hardening any entity not in $KillSet(S)$ does not provide any benefit to the network operator. In this study we also assume that the set of $\cal{K}$ most vulnerable entities in the network is \emph{unique}.

\begin{table}[H]
\begin{center}
	\begin{subtable}[t]{0.45\columnwidth}
		\begin{tabular}{|c|c|c|c|c|c|}  \hline
		\multicolumn{1}{|c|}{Entities} & \multicolumn{5}{c|}{Time Steps ($t$)}\\
		\cline{2-6} & $0$ & $1$ & $2$ & $3$ & $4$ \\\hline \hline
		$a_1$ & $0$ & $0$ & $0$ & $0$ & $0$ \\ \hline
		$a_2$ & $*$ & $*$ & $*$ & $*$ & $*$ \\ \hline
		$a_3$ & $0$ & $0$ & $0$ & $0$ & $0$ \\ \hline
		$a_4$ & $0$ & $0$ & $0$ & $0$ & $0$ \\ \hline
		$b_1$ & $0$ & $0$ & $0$ & $0$ & $0$ \\ \hline
		$b_2$ & $0$ & $0$ & $0$ & $0$ & $0$ \\ \hline
		$b_3$ & $1$ & $1$ & $1$ & $1$ & $0$ \\ \hline		
		\end{tabular}
		\caption{Entity $a_2$ is hardened}\label{tbl:example1modcascade_a}
	\end{subtable}
	\quad
	\begin{subtable}[t]{0.45\columnwidth}
		\begin{tabular}{|c|c|c|c|c|c|}  \hline
		\multicolumn{1}{|c|}{Entities} & \multicolumn{5}{c|}{Time Steps ($t$)}\\
		\cline{2-6} & $0$ & $1$ & $2$ & $3$ & $4$ \\\hline \hline
		$a_1$ & $0$ & $0$ & $1$ & $1$ & $1$ \\ \hline
		$a_2$ & $1$ & $1$ & $1$ & $1$ & $1$ \\ \hline
		$a_3$ & $0$ & $0$ & $0$ & $0$ & $0$ \\ \hline
		$a_4$ & $0$ & $0$ & $0$ & $0$ & $0$ \\ \hline
		$b_1$ & $0$ & $0$ & $0$ & $1$ & $1$ \\ \hline
		$b_2$ & $0$ & $1$ & $1$ & $1$ & $1$ \\ \hline
		$b_3$ & $*$ & $*$ & $*$ & $*$ & $*$ \\ \hline		
		\end{tabular}
		\caption{Entity $b_3$ is hardened}\label{tbl:example1modcascade_b}
	\end{subtable}
	\vspace{-10pt}
	\caption{Failure cascade propagation with entity hardening. Entities $\{a_2, b_3\}$ are attacked at time step $t=0$. A value of $1$ denotes entity failure, $0$ otherwise. $*$ denotes a hardened entity.}
\protect\label{tbl:example1modcascade}
\end{center}
\vspace{-25pt}
\end{table}

We now proceed to formulate the entity hardening problem formally. Given an interdependent network system $\mathcal{I}(A,B,\mathcal{F}(A,B))$, and the set of $\cal{K}$ most vulnerable entities of the system $A' \cup B'$, where $A' \subseteq A$ and $B' \subseteq B$:

\vspace{0.05in}
\noindent
\textbf{{\em The Entity Hardening (ENH) problem}}\\
INSTANCE: Given:\\
(i) An interdependent network system $\mathcal{I}(A,B,\mathcal{F}(A,B))$, where the sets $A$ and $B$ represent the entities of the two networks, and $\mathcal{F}(A,B)$ is the set of IDRs.\\
(ii) The set of $\cal{K}$ most vulnerable entities of the system $A' \cup B'$, where $A' \subseteq A$ and $B' \subseteq B$ \\
(iii) Two positive integers $k, k<\cal{K}$ and $E_F$. \\ \\
QUESTION:Is there a set of entities $\mathcal{H} = A'' \cup B'', A'' \subseteq A, B'' \subseteq B, |\mathcal{H}| \leq k$, such that hardening $\mathcal{H}$ entities results in no more than $E_F$ entities to fail after entities $A' \cup B'$ fail at time step $t=0$. 

We note some of the assumptions for the ENH problem: First, we assume that once an entity is hardened, it is always operational and does not fail at any time step of the observation, even when the entity is part of the $\cal{K}$ most vulnerable entities. Second, we assume that $k<\cal{K}$, as otherwise the selection of $\cal{K}$ entities for hardening ensures that no entities fail at all. Finally, as noted earlier, we assume that the set of $\cal{K}$ most vulnerable entities in the network is \emph{unique}. We now proceed to analyze the computational complexity of the ENH problem.



\section{Computational Complexity Analysis}
\label{CompAna}
For an interdependent network $\mathcal{I}(A,B,\mathcal{F}(A,B))$ the IDRs can be represented in four different forms. We analyze the computational complexity of the ENH problem for each of these cases separately.

\subsection{Case I: Problem Instance with One Minterm of Size One}
The IDRs of Case I have a single minterm of size $1$. This can be represented as $x_{i} \leftarrow y_{j}$, where $x_{i}$ and $y_{j}$ are entities of network $A(B)$  and $B(A)$ respectively. We show that the ENH problem for Case I can be solved optimally in polynomial time. 



\begin{algorithm}
\small
	\KwData{An interdependent network $\mathcal{I}(A,B,\mathcal{F}(A,B))$, set of $\cal{K}$ most vulnerable entities $A' \cup B', A' \subseteq A, B' \subseteq B$, hardening budget $k$ and a set $\mathcal{H}=\emptyset$. 
		}
	\KwResult{Set of hardened entities $\mathcal{H}$.
		}
	\Begin{			
		  For each entity $x_i \in (A' \cup B')$ compute the set of kill sets $\mathcal{C}=\{C_{x_1}, C_{x_2}, ..., C_{x_{ \mathcal{K}}}\}$, where $C_{x_i} = KillSet(x_i)$ \;
			Create a copy $\mathcal{D}=\{D_{x_1}, D_{x_2}, ..., D_{x_{\mathcal{K}}}\}$ of set $\mathcal{C}$ \;		
			\For  {(i=1; $i \le \mathcal{K}$; i++)}{
     				\For  {(j=1, $j \ne i$; $j \le \mathcal{K}$; j++)}{
					\If{$C_{x_j} \subset C_{x_i}$}{
						$D_{x_i} \leftarrow D_{x_i} \setminus D_{x_j}$ \;
					}
				}
			}
			Choose the top $k$ sets from $\mathcal{D}$ with highest cardinality \;			
			For each of the $D_{x_i} \subseteq \mathcal{D}$ sets chosen in Step 8, $\mathcal{H} \leftarrow \mathcal{H} \cup x_i$ \;
			\Return{$\mathcal{H}$}
	}		
\caption{Entity Hardening Algorithm for systems with Case I type interdependencies}
\label{algCaseI}
\end{algorithm}


\begin{theorem}{} \label{CaseI} Algorithm \ref{algCaseI} solves the Entity Hardening problem for Case I optimally in polynomial time. 
\end{theorem}
\vspace{-0.091in}
\begin{proof}
It is shown in \cite{sen2014identification} that the kill set for all entities in the interdependent network can be computed in $\mathcal{O}(n^3)$ where $n=|A|+|B|$, thus computing the kill sets for $\mathcal{K}$ entities takes $\mathcal{O}(\mathcal{K}n^2)$. Step 4-7 of the algorithm runs in $\mathcal{O}(\mathcal{K}^2)$. Choosing the $k$ highest cardinality sets can be found using any standard sorting algorithm in $\mathcal{O}(\mathcal{K}log(\mathcal{K}))$. Hence Algorithm \ref{algCaseI} runs in $\mathcal{O}(\mathcal{K}n^2)$. 


For two kill sets $C_{x_i}$ and $C_{x_j}$ it can be shown that either $C_{x_i} \cap C_{x_j} = \emptyset$ or $C_{x_i} \cap C_{x_j} = C_{x_i}$ or $C_{x_i} \cap C_{x_j} = C_{x_j}$ \cite{sen2014identification}. So with two entities $\{x_i , x_j\} \in A' \cup B'$ and $C_{x_i} \cap C_{x_j} = C_{x_j}$ i.e, $C_{x_j} \subset C_{x_i}$, if $x_i$ is hardened it prevents the failure of  $C_{x_i} - C_{x_j}$ entities (provided that none of the entities in $C_{x_i} - C_{x_j} - \{x_i\}$ are in  $A' \cup B'$). With this assertion, for an entity $x_i \in A' \cap B'$, steps 4-7 of Algorithm \ref{algCaseI} finds the actual entities for which failure is prevented by hardening $x_i$. The set $\mathcal{D}=\{D_{x_1}, D_{x_2}, ..., D_{x_{\mathcal{K}}}\}$ comprises of these set of entities for each hardened entity $x_i$. 


To prove that Algorithm \ref{algCaseI} finds the optimal solution we make the following two assertions: First, consider any two sets $D_{x_i}$ and $D_{x_j}$. It is implied from step 6 of Algorithm \ref{algCaseI} that  $D_{x_i} \cap D_{x_j} = \emptyset$. Second, consider an entity $x_p \notin A' \cup B'$ is hardened. If $x_p$ fails when entities in $A' \cup B'$ fails initially then it would belong to some set $D_{x_i}$. Thus hardening $x_p$ results in preventing the failure of entities that is a proper subset of $D_{x_i}$. Hence the entities to be hardened must belong to $A' \cup B'$ only. Owing to the two assertions it directly follows that with a given budget $k$, hardening $k$ highest cardinality sets from the set $\mathcal{D}$ ensures prevention of failure for the maximum number of entities. 
\end{proof}

\subsection{Case II: Problem Instance with One Minterm of Arbitrary Size}
The IDRs of Case II have a single minterm of arbitrary size. This can be represented as $x_{i} \leftarrow \prod^{p}_{j=1} y_{j}$, where $x_{i}$ and $y_{j}$ are entities of network $A(B)$  and $B(A)$ respectively and the size of the minterm is $p$. The Entity Hardening problem with respect to Case II is NP-complete and is proved in Theorem \ref{CaseII}.  An inapproximability proof for this case of the problem is given in Theorem \ref{inapxC2}

\begin{theorem}{} \label{CaseII}The Entity Hardening problem for Case II is NP Complete 
\end{theorem}
\begin{proof}
The Entity Hardening problem for case II is proved to be NP complete by giving a reduction from the Densest $p$-Subhypergraph problem \cite{hajiaghayi2006minimum}, a known NP-complete problem. An instance of the  Densest $p$-Subhypergraph problem includes a hypergraph $G=(V,E)$, a parameter $p$ and a parameter $M$. The problem asks the question whether there exists a set of vertices $|V'| \subseteq V$ and $|V'| \le p$ such that the subgraph induced with this set of vertices has at least $M$ hyperedges. From an instance of the  Densest $p$-Subhypergraph problem we create an instance of the ENH problem in the following way. For each vertex $v_i$ and each hyperedge $e_j$ an entity $b_i$ and $a_j$ are added to the set $B$ and $A$ respectively. For each hyperedge $e_j$ with $e_j= \{v_m, v_n, v_q\}$ (say) an IDR of form $a_j \leftarrow b_m b_n b_q$ is created. It is assumed that the value of $\mathcal{K}$ is set of $|V|$. The values of $k$ and $E_F$ are set to $p$ and $|V|+|E|-p-M$ (where $|A|=|V|$ and $|B|=|E|$) respectively.

In the constructed instance only entities of set $A$ are dependent on entities of set $B$. Additionally the dependency for an entity $a_i$ consists of conjunction of entities in set $B$. Hence for an entity $a_i \in A$ to fail, either it itself has to fail initially or all entities to which $a_i$ is dependent on has to fail. It is to be noted that the entities in set $B$ has no induced failure i.e., there is no cascade. Following from this assertion, with $\mathcal{K}=p$, the solution $A'=\emptyset$ and $B'=B$ would fail all entities in set $A \cup B$. Moreover this is the single unique solution to the problem instance. This is because by including one entity $a_i$ in the initial failure set would result in not failing at least one entity $b_j$ for a given budget $\mathcal{K}=p$. Hence it won't fail the entire set of entities in $A \cup B$. 

If an entity in set $A$ is hardened then it would have no effect in failure prevention of any other entities. Whereas hardening an entity $b_m \in B$ might result in failure prevention of an entity $a_i \in A$ with IDR $a_j \leftarrow b_m b_n b_q$ provided that entities $b_n, b_q$ are also defended. With $k = p$ (and $\mathcal{K} \le |V|=|B|$) it can be ensured that entities to be defended are from set $B'$. 

To prove the theorem consider that there is a solution to the Densest $p$-Subhypergraph problem. Then there exist $p$ vertices which induces a subgraph which has at least $M$ hyperedges. Hardening the entities $b_i \in B'$ for each vertex $v_i$ in the solution of the Densest $p$-Subhypergraph problem would then ensure that at least $M$ entities in set $A$ are protected from failure. This is because the entities in set $A$ for which the failure is prevented corresponds to the hyperedges in the induced subgraph. Thus the number of entities that fail after hardening $p$ entities is at most $|V|+|E|-p-M$, solving the ENH problem. Now consider that there is a solution to the ENH problem. As previously stated, the entities to be hardened will always be from set $B'$. So defending $p$ entities from set $B'$ would result in failure prevention of at least $M$ entities in set $A$ such that $E_F \le |V|+|E|-p-M$. Hence, the vertex induced subgraph would have at least $M$ hyperedges when vertices corresponding to the entities hardened are included in the solution of the Densest $p$-Subhypergraph problem, thus solving it. 
\end{proof}
 
\begin{theorem} {} \label{inapxC2} For an interdependent network $\mathcal{I}(A,B,\mathcal{F}(A,B))$ with $n=|A \cup B|$ and $\mathcal{F}(A,B)$ having IDRs of form Case II, it is hard to approximate the ENH problem within a factor of $\frac{1}{2^{log(n)^{\lambda}}}$ for some $\lambda >0$. 
\end{theorem}

\begin{proof}
From Theorem \ref{CaseII}, Densest $p$-Subhypergraph problem has been shown to be a special case of the ENH problem with IDRs of form Case II.  Densest $p$-Subhypergraph problem is proved to be inapproximable  within a factor of $\frac{1}{2^{log(n)^{\lambda}}}$ ($\lambda >0$) in \cite{hajiaghayi2006minimum}. Hence the theorem follows.
\end{proof}

\subsection{Case III: Problem Instance with an Arbitrary Number of Minterm of Size One}
The IDRs of Case III have arbitrary number of minterm of size $1$. This can be represented as $x_{i} \leftarrow \sum^{p}_{q=1} y_{q}$, where $x_{i}$ and $y_{q}$ are entities of network $A(B)$  and $B(A)$ respectively and the number of minterms are $p$. The ENH problem with respect to Case III is NP-complete and is proved in Theorem \ref{CaseIII}.
\begin{theorem}{} \label{CaseIII}The ENH problem for Case III is NP Complete 
\end{theorem}

\begin{proof}
The ENH problem for case III is proved to be NP complete by giving a reduction from the Set Cover Problem, a well known NP-complete problem. An instance of the Set Cover problem includes a set $S=\{x_1,x_2,...,x_n\}$, a set $\mathcal{S}=\{S_1,S_2,...,S_m\}$ where $S_i \subseteq S$ and a positive integer $M$. The problem asks the question whether there exists at most $M$ subsets from set $\mathcal{S}$ whose union would result in the set $S$. From an instance of the set cover problem we create an instance of the ENH problem in the following way. For each element $x_i$ in set $S$ we add an entity $a_i$ in set $A$. For each subset $S_i$ in set $\mathcal{S}$ we add an entity $b_i$ in set $B$. For all subsets in $\mathcal{S}$, say $S_p, S_m,S_n$, which has the element $x_i$ there is an IDR of form $a_i \leftarrow b_m+b_n+b_l$.  The values of positive integers $k$ and $E_F$ are set to $M$ and $m-M$ respectively. It is assumed that the value of $\mathcal{K}=m$.

With similar reasoning as that of Case II it can be shown that for $\mathcal{K}=m$ the maximum number of node failures (i.e. failure of all entities in $A \cup B$) would occur if $A'=\emptyset$ and $B'=B$. This is also the single unique solution to the problem instance. 

The constructed instance also ensures that the entities to be hardened are from set $B'$ ($A'$ not considered as it is equal to $\emptyset$). This is because protecting an entity $a_i \in A$ would only result in prevention of its own failure whereas protecting an entity $b_j \in B$ would result in failure prevention of its own and all other entities in set $A$ for which it appears in its IDR. 

To begin with the proof, consider that there is a solution to the Set Cover problem. Then there exist $M$ subsets (or elements in set $\mathcal{S}$) whose union results in the set $S$. Hardening the entities in set $B$ corresponding to the subsets selected would ensure that all entities in set $A$ are prevented from failure. This is because for the dependency of each entity $a_i \in A$ there exist at least one entity (in set $B$) that is hardened. Hence the number of entities that fails after hardening is $m-M$ which is equal to $E_F$, thus solving the ENH problem. Now, consider that there is a solution to the ENH problem. As discussed above the entities to be hardened should be from set $B'$. To achieve $E_F=m-M$ with $k=M$, no entities in the set $A$ must fail. Hence for each entity $a_i \in A$ at least one entity in set $B$ that appears in its IDR has to be hardened. Thus, it directly follows that the union of subsets in set $\mathcal{S}$ corresponding to the entities hardened is equal to the set $S$, solving the Set Cover Problem. 
\end{proof}

\subsubsection{Approximation Scheme for Case 3}
In this subsection we provide an approximation algorithm for Case 3 of the problem. For an interdependent network $\mathcal{I}(A,B,\mathcal{F}(A,B))$ with the initial failed set of entities as $A' \cup B'$ we define \emph{Protection Set} of each entity as follows. 
\begin{definition}
\textit{For an entity $x_i \in A \cup B$ the Protection Set is defined as the entities that would be prevented from failure by hardening the entity $x_i$ when all entities in $A' \cup B'$ fails initially. This is represented as $P(x_i|A' \cup B')$.}
\end{definition}

The Protection Set of each entity can be computed in $\mathcal{O}((n+m)^2)$ where $n$ and $m$ are the number of entities and number of minterms respectively in an interdependent network $\mathcal{I}(A,B,\mathcal{F}(A,B))$ .
\begin{theorem}
\label{CaseIIIappx1}
For two entities $x_i,x_j \in A \cup B$, $P(x_i|A' \cup B') \cup P(x_j|A' \cup B') = P(x_i,x_j|A' \cup B')$ when IDRs are of form Case III.
\end{theorem}
\begin{proof}
Assume that defending two entities $x_i$ and $x_j$ would result in preventing failure of $P(x_i,x_j|A' \cup B')$ entities with $|P(x_i|A' \cup B') \cup P(x_j|A' \cup B')| < |P(x_i,x_j|A' \cup B')|$. Then there exist at least one entity $x_p \notin P(x_i|A' \cup B') \cup P(x_j|A' \cup B')$ such that it's failure is prevented only if $x_i$ and $x_j$ is protected together. So two entities $x_m$ and $x_n$ (with $x_m \in P(x_i|A' \cup B')$ and $x_n \in P(x_j|A' \cup B')$ or vice versa) have to be present in the IDR of $x_p$. As the IDRs are of form Case III so if any one of $x_m$ or $x_n$ is protected then $x_p$ is protected, hence a contradiction. On the other way round $P(x_i,x_j|A' \cup B')$ contains all entities which would be prevented from failure if $x_i$ or $x_j$ is defended alone. So it directly follows that $|P(x_i|A' \cup B') \cup P(x_j|A' \cup B')| > |P(x_i,x_j|A' \cup B')|$ is not possible. Hence the theorem holds. 
\end{proof}

\begin{theorem}
\label{CaseIIIappx2}
There exists an $1-\frac{1}{e}$ approximation algorithm that approximates the ENH problem for Case III.
\end{theorem}
\begin{proof}
The approximation algorithm is constructed by modeling the problem as Maximum Coverage problem. An instance of the maximum coverage problem consists of a set $S=\{x_1,x_2,...,x_n\}$, a set $\mathcal{S}=\{S_1,S_2,...,S_m\}$ where $S_i \subseteq S$ and a positive integer $M$. The objective of the problem is to find a set $S' \subseteq S$ and $|S'| \le M$ such that $\cup_{S_i \in \mathcal{S}} S_i$ is maximized. For a given initial failure set $A' \cup B'$ with $|A'| + |B'| \le \mathcal{K}$, let $P(x_i|A' \cup B')$ denote the protection set for each entity $x_i \in A \cup B$. We construct a set $S=A \cup B$ and for each entity $x_i$ a set $S_{x_i} \subseteq S$ such that $S_{x_i} =  P(x_i|A' \cup B')$. Each set $S_{x_i}$ is added as an element of a set $\mathcal{S}$. The conversion of the problem to Maximum Coverage problem can be done in polynomial time. By Theorem \ref{CaseIIIappx1} defending a set of entities $X \subseteq S$ would result in failure prevention of $\cup_{x_i \in X} S_{x_i}$ entities. Hence, with the constructed sets $S$ and $\mathcal{S}$ and a positive integer $M$ (with $M=k$) finding the Maximum Coverage would ensure the failure protection of maximum number of entities in $A \cup B$. This is same as the ENH problem of Case III. As there exists an $1-\frac{1}{e}$ approximation algorithm for the Maximum Coverage problem hence the theorem holds. 
\end{proof}
\subsection{Case IV: Problem Instance with an Arbitrary Number of Minterms of Arbitrary Size}
The IDRs of Case IV have arbitrary number of minterm of arbitrary size. This can be represented as $x_{i} \leftarrow \sum^{p}_{j_1=1} \prod^{q_{j_1}}_{j_2=1}y_{j_2}$, where $x_{i}$ and $y_{j_2}$ are entities of network $A(B)$ and $B(A)$ respectively and there are $p$ minterms each of size $q_{j_1}$. 

\begin{theorem}\label{CaseIV} The Entity Hardening problem for Case IV is NP Complete 
\end{theorem}
\begin{proof}
Case II and Case III are special cases of Case IV. Hence following from Theorem \ref{CaseII} and Theorem \ref{CaseIII} the computational complexity of the Entity Hardening problem is NP-complete in Case IV.
\end{proof}
\section{Solutions to the Entity Hardening Problem}
\label{Solutions}
\subsection{Optimal Solution using Integer Linear Programming}

We propose an Integer Linear Program (ILP) that solves the Entity Hardening problem optimally. Let $[G,H]$ with $G=\{g_1,g_2,...,g_n\}$ and $H=\{h_1,h_2,...,h_m\}$ denote the entities in set $A$ and $B$ respectively with $h_i=0$ ($g_j=0$) if entity $a_i$ ($b_j$) is alive and $h_i=1$ ($g_j=1$) otherwise. Given an integer $k$ let $[G,H]$ be the solution (with value of $1$ corresponding to entities failed initially) that cause maximum number of entity failure. Two variables $x_{id}$ and $y_{jd}$ are used in the ILP with $x_{id} = 1$ ($y_{jd} = 1$), when entity $a_{i} \in A$ ($b_{j} \in B$) is in a failed state at time step $d$, and $0$ otherwise. The number of entities to be defended is considered to be $k$. It is to be noted that the maximum number cascading steps is upper bounded by $|A|+|B|-1=m+n-1$. The objective function can now be formulated as follows:\\
\noindent
\begin{equation}\label{eqn:ilpobj1}
min \Big(\overset{m}{\underset{i=1}{\sum}}x_{i(m+n-1)}+\overset{n}{\underset{j=1}{\sum}}y_{j(m+n-1)} \Big)
\end{equation}

\noindent
The objective in (\ref{eqn:ilpobj1}) minimizes the number of entities failed after the cascading failure with the respective constraints for the Entity Hardening problem as follows:

\noindent
{\em Constraint Set 1}: $\sum\limits_{i=1}^{n} q_{x_i} + \sum\limits_{j=1}^{m} q_{y_j} = k$ , with $q_{xi},q_{yj} \in [0,1]$. If an entity $x_i$ ($y_j$) is defended then $q_{x_i} = 1$ ($q_{y_j} =1$) and $0$ otherwise. \\ \\
\noindent
{\em Constraint Set 2:} $x_{i0} \ge g_{i} - q_{x_i}$ and $y_{i0} \ge h_{i} - q_{y_i}$. This constraint implies that only if an entity is not defended and $g_i$ ($h_i$) is $1$ then the entity will fail at the initial time step.\\
\noindent \\
{\em Constraint Set 3}: $x_{id} \geq x_{i(d-1)}, \forall d, 1 \leq d \leq m+n-1$, and  $y_{id} \geq y_{i(d-1)}, \forall d, 1 \leq d \leq m+n-1$, in order to ensure that for an entity which fails in a particular time step would remain in failed state at all subsequent time steps.\\
\noindent \\
{\em Constraint Set 4}: Modeling of the constraint to capture the cascade propagation for IIM is similar to the constraints established in \cite{sen2014identification}. A brief presentation of this constraint is provided here. Consider an IDR ${a_i} \leftarrow {b_j}{b_p}{b_l} + {b_m}{b_n} + {b_q}$ of type Case IV. The following steps are enumerated to depict the cascade propagation:

\vspace{0.02in}
\noindent
{\em Step 1:} Replace all minterms of size greater than one with a variable. In the example provided we have the transformed minterm as ${a_i} \leftarrow c_1 + c_2 + b_q$ with $c_1 \leftarrow {b_j}{b_p}{b_l}$ and $c_2 \leftarrow {b_m}{b_n}$ ($c_1,c_2 \in \{0,1\}$) as the new IDRs. Note that after transformation, the original IDR is in the form of Case III and the introduced IDRs are in the form of Case II.

\vspace{0.02in}
\noindent
{\em Step 2:} For each variable $c$, a constraints is added to capture the cascade propagation. Let $N$ be the number of entities in the minterm on which $c$ is dependent. In the example for the variable $c_1$ with IDR $c_1 \leftarrow {b_j}{b_p}{b_l}$, constraints $c_{1d} \geq \frac{y_{j(d-1)} + y_{p(d-1)}+y_{l(d-1)}}{N}$ and $c_{1d} \le y_{j(d-1)} + y_{p(d-1)}+y_{l(d-1)} \forall d, 1 \leq d \leq m+n-1$ are introduced (with $N=3$ in this case). If IDR of an entity is already in form of Case II, i.e.,$a_i \leftarrow {b_j}{b_p}{b_l}$ then constraints $x_{id} \geq \frac{y_{j(d-1)}+y_{p(d-1)}+ y_{l(d-1)}}{N} - q_{x_i}$ and $x_{id} \le y_{j(d-1)} + y_{p(d-1)}+y_{l(d-1)}  \forall d, 1 \leq d \leq m+n-1$ are introduced (with $N=3$). These constraints satisfies that if the entity $x_i$ is hardened initially then it is not dead at any time step. 

\vspace{0.02in}
\noindent
{\em Step 3:} Let $M$ be the number of minterms in the transformed IDR as described in Step 1. In the given example with IDR ${a_i} \leftarrow c_1 + c_2 + b_q$  constraints of form $x_{id} \geq c_{1(d-1)} + c_{2(d-1)} + y_{q(d-1)}-(M-1)-q_{x_i}$ and $x_{id} \leq \frac{c_{1(d-1)} + c_{2(d-1)} + y_{q(d-1)}}{M} \forall d, 1 \leq d \leq m+n-1$ are introduced. These constraints ensures that even if all the minterms of $x_i$ has at least one entity in dead state then it will be alive if the entity is hardened initially. For all IDRs of type Case I and Case III, the constraint discussed in this step is used.

\subsection{Heuristic}
In this subsection we provide a greedy heuristic solution to the Entity Hardening problem. For an interdependent network $\mathcal{I}(A,B,\mathcal{F}(A,B))$ with the initial failed set of entities as $A' \cup B'$, \emph{Protection Set} of each entity has been defined in the approximation scheme of Case III. To design the heuristic we define \emph{Minterm Coverage Number} of each entity in $A \cup B$ as follows:
\begin{definition}
\textit{For an entity $x_i \in A \cup B$ the Minterm Coverage Number is defined as the number of minterms that can be removed from $\mathcal{F}(A,B)$ without affecting the cascading process by hardening the entity $x_i$ when all entities in $A' \cup B'$ fails initially. This is represented as $M(x_i|A' \cup B')$.}
\end{definition}

Similar to the computation of \emph{Protection Set} the \emph{Minterm Coverage Number} of each entity can be computed in $\mathcal{O}((n+m)^2)$. With these definitions the heuristic is given in Algorithm \ref{algHeu}. The algorithm takes in as input an interdependent network $\mathcal{I}(A,B,\mathcal{F}(A,B))$ with $S=A \cup B$. Step 4-5 is done to reduce the search space as it directly follows that the set of entities in $Q$ wouldn't effect the hardening process. In each iteration of the while loop an entity $x_d$ is greedily selected which when hardened would prevent failure of maximum number of entities. This ensures that at each step the number of entities failed is minimized. In case of a tie, among all entities involved in the tie, the entity having the highest Minterm Coverage Number is included in the solution. This gives a higher priority to the entity which when hardened, has more impact on failure minimization in subsequent iterations of the while loop. The interdependent network $\mathcal{I}(A,B,\mathcal{F}(A,B))$ is updated in steps 13-16 of the algorithm. This takes into account the effect of hardening an entity in the current iteration on entities hardened in the following iterations. 

\begin{algorithm}
\small	
	\KwData{An interdependent network $\mathcal{I}(A,B,\mathcal{F}(A,B))$ (with $S=A \cup B$), set of entities $A' \cup B'$ failed initially causing maximum failure in the interdependent network with $|A'|+|B'| = \mathcal{K}$ and hardening budget $k$. 		
		}
	\KwResult{Set of hardened entities $\mathcal{H}$.
		}
	\Begin{	
			Initialize $S' \leftarrow A' \cup B'$ \;
			Initialize $\mathcal{H}=\emptyset$\;
			Update $\mathcal{F}(A,B)$ as follows --- (a) let $Q$ be the set of entities that does not fail on failing 	                                $\mathcal{K}$ entities, (b) remove IDRs corresponding to entities in set $Q$, (c) remove from minterm  
                                of entities not in set $Q$ all entities which are in set $Q$ \;
			Update $S =S \setminus Q$ \;
			\While {($k$ entities are not hardened)} {
			For each entity $x_i \in S$ compute the Protection Set $P(x_i|S')$\;	
			Choose the entity $x_d$ with highest cardinality of the set $|P(x_d|S')|$\;
			\If {(more than one entity has the same highest cardinality value)} {
				For each such entity $x_j$ compute the Minterm Coverage Number $M(x_j|S')$ \;
				Choose the entity $x_d$ with highest Minterm Coverage Number. \; 
				In case of a tie choose arbitrarily\;
				}
			Update $S \leftarrow S-P(x_d|S')$\;
			Update $\mathcal{F}(A,B)$ by removing (i) IDRs corresponding to all entities in $P(x_d|S')$, and (ii) occurrence of these entities in IDRs of entities not in $P(x_d|S')$\;
			\If {($x_d \in S'$)} {
				Update $S' \leftarrow S'-\{x_d\}$\;
				}
			Update $\mathcal{H}=\mathcal{H} \cup x_d$ \;
			}	
	}
	\textbf{return} $\mathcal{H}$ \;		
\caption{Heuristic Solution to the ENH Problem}
\label{algHeu}
\end{algorithm}

\noindent
\emph{Run Time Analysis of Algorithm \ref{algHeu}}: For this analysis we consider $n$ to be the total number of entities and $m$ to be the total number of minterms. Updates in step 4 can be done in $\mathcal{O}(m)$ and step 5 in $\mathcal{O}(n)$. The while loop iterates for $k$ times. In each iteration of the while loop step 7 and step 8 takes at most $\mathcal{O}((n+m)^2)$ and $\mathcal{O}(n log(n))$ time respectively. On branching in step 9, step 10 and step 11 takes $\mathcal{O}((n+m)^2)$ and $\mathcal{O}(n log(n))$ time respectively. Updates in step 13 takes $\mathcal{O}(n)$ time and in step 14 takes $\mathcal{O}(n+m)$ time. Step 12, 16 and 17 runs in constant time. Hence Algorithm \ref{algHeu} runs in $\mathcal{O}(k(n+m)^2)$ time. 
\begin{figure*}[!htb]
	\centering
	\begin{subfigure}[t]{0.6\columnwidth}
		\centering
		\includegraphics[width=\columnwidth, keepaspectratio]{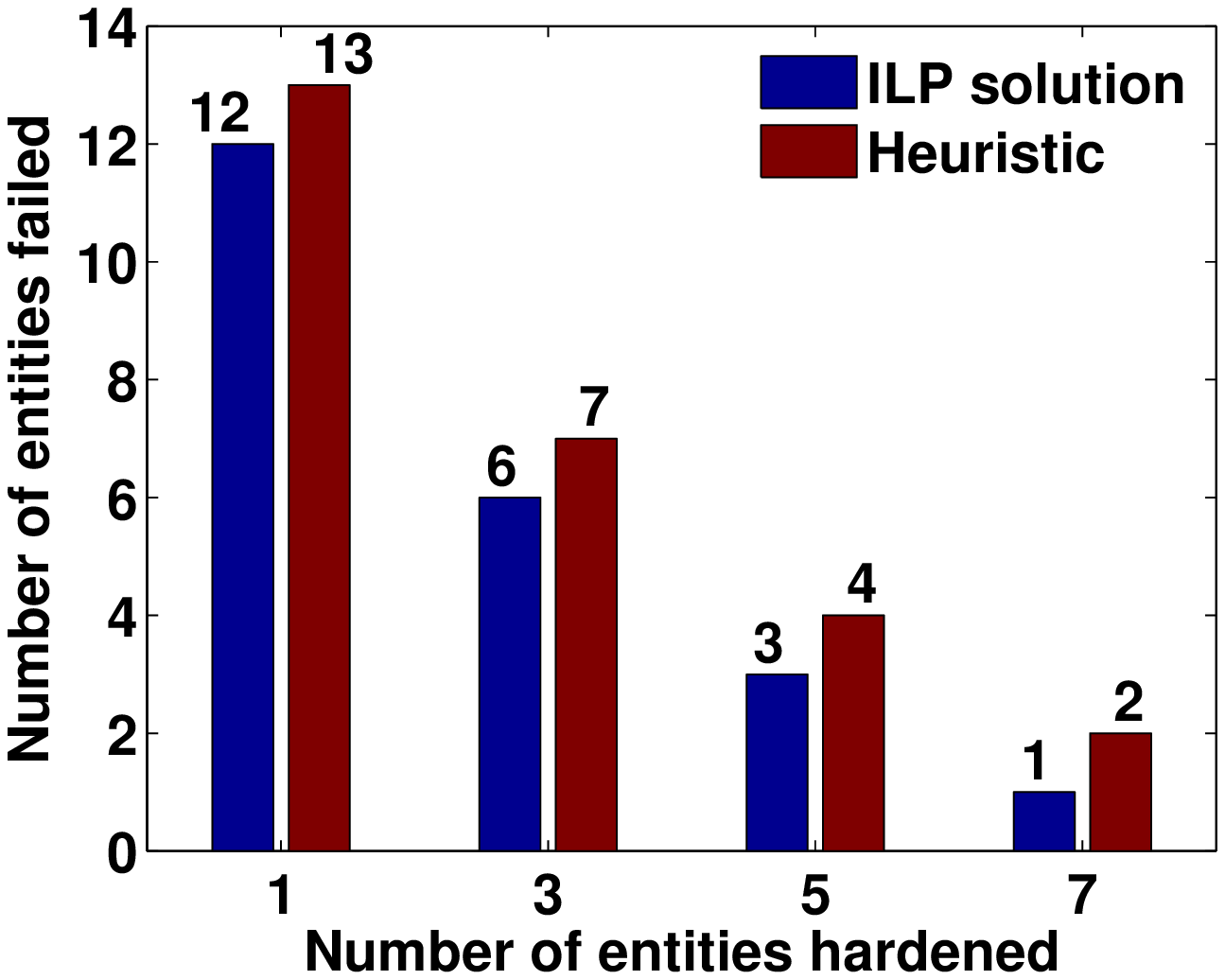}
		\caption{Region 1}\label{fig:plota}
	\end{subfigure}	
	\begin{subfigure}[t]{0.6\columnwidth}
		\centering
		\includegraphics[width=\columnwidth, keepaspectratio]{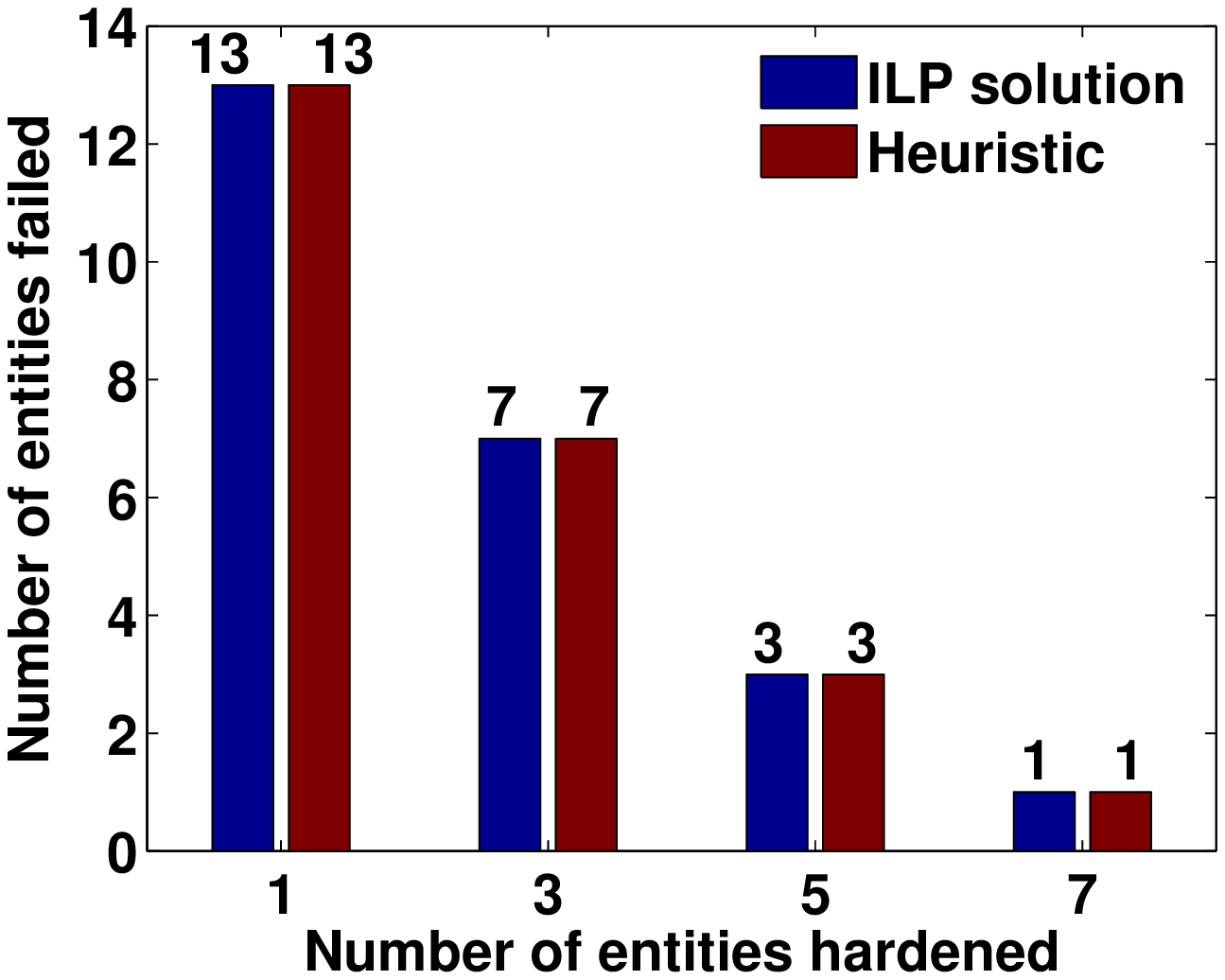}
		\caption{Region 2}\label{fig:plotb}
	\end{subfigure}
	\begin{subfigure}[t]{0.6\columnwidth}
		\centering
		\includegraphics[width=\columnwidth, keepaspectratio]{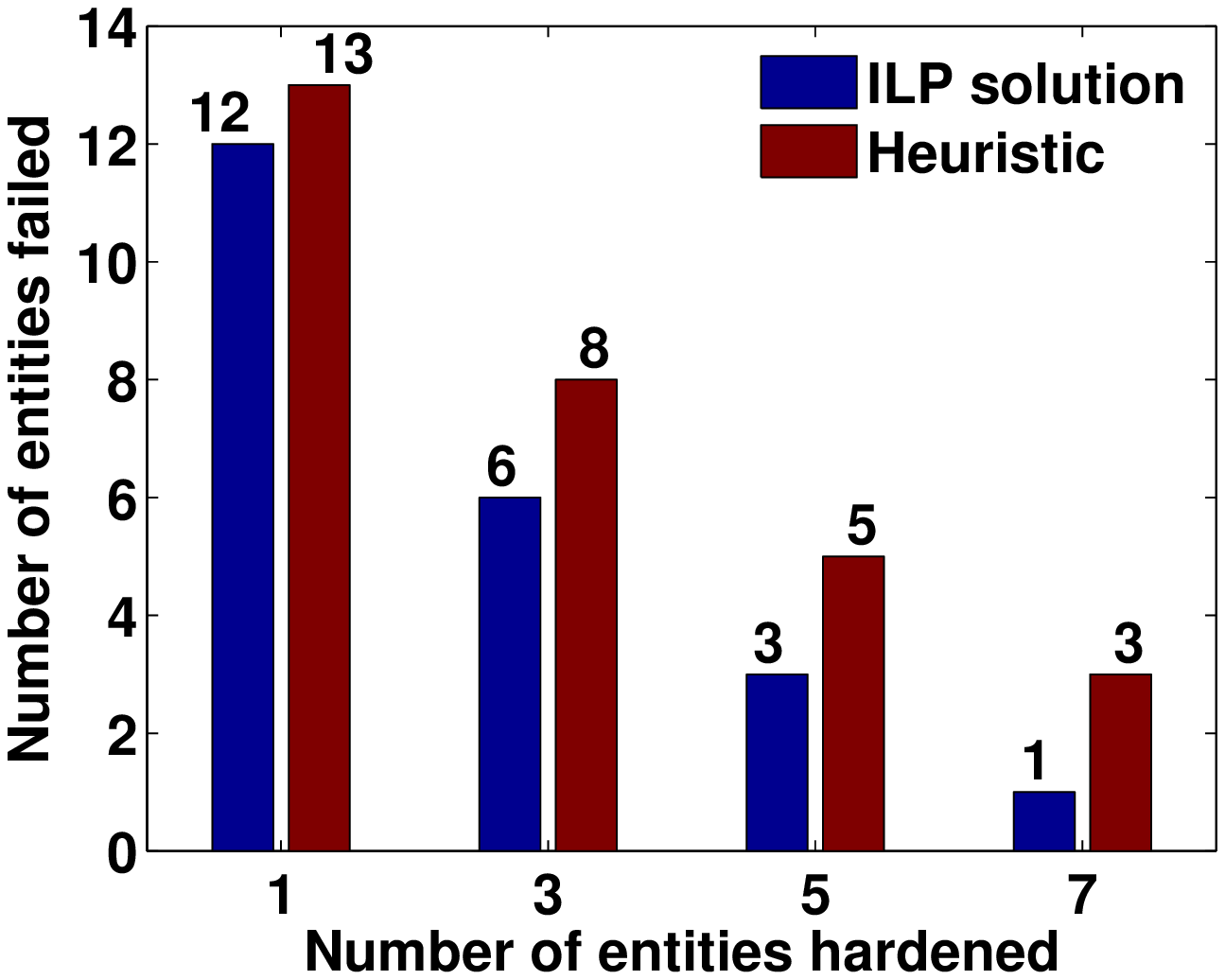}
		\caption{Region 3}\label{fig:plotc}
	\end{subfigure}
	\begin{subfigure}[t]{0.6\columnwidth}
		\centering
		\includegraphics[width=\columnwidth, keepaspectratio]{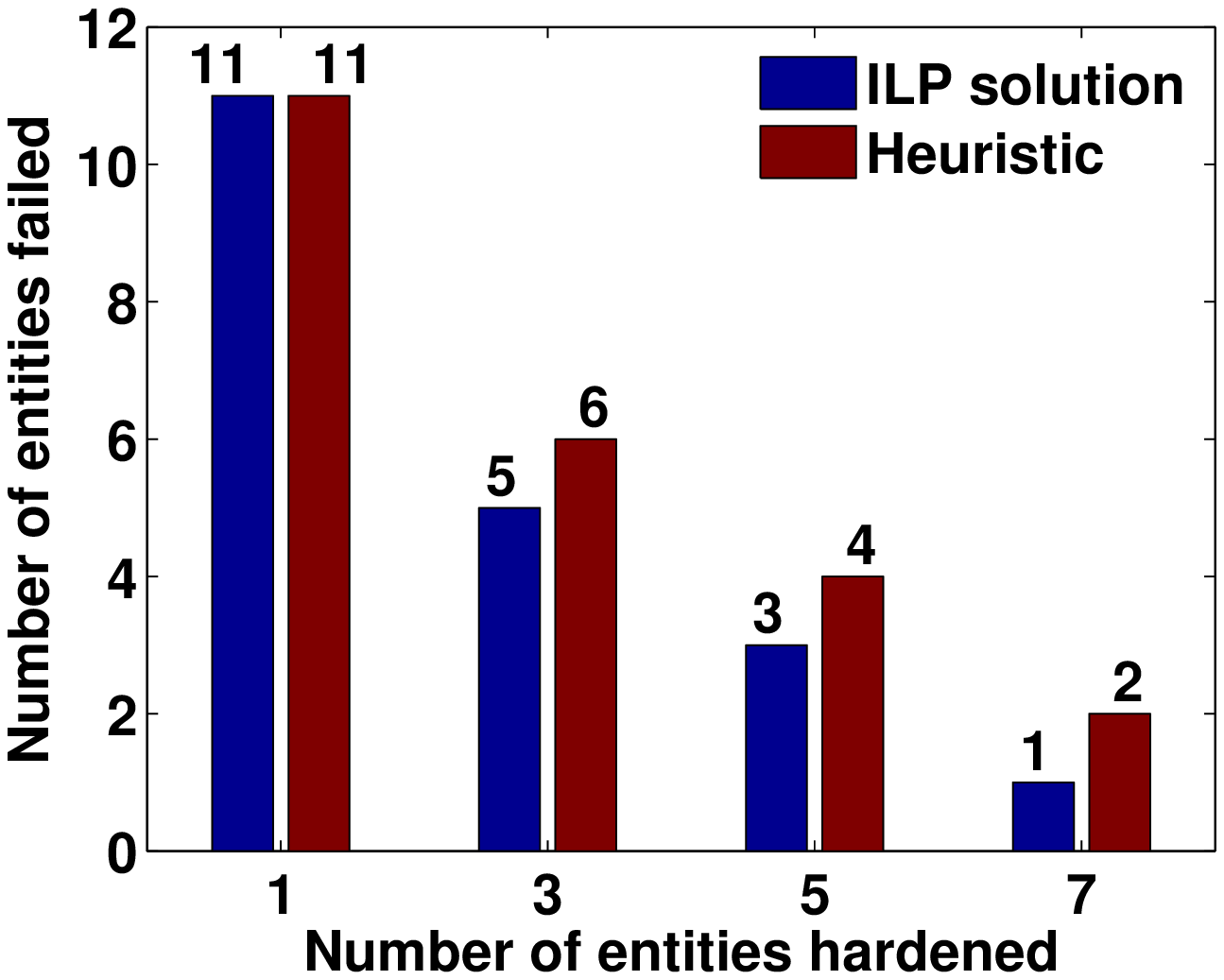}
		\caption{Region 4}\label{fig:plotd}
	\end{subfigure}
	\begin{subfigure}[t]{0.6\columnwidth}
		\centering
		\includegraphics[width=\columnwidth, keepaspectratio]{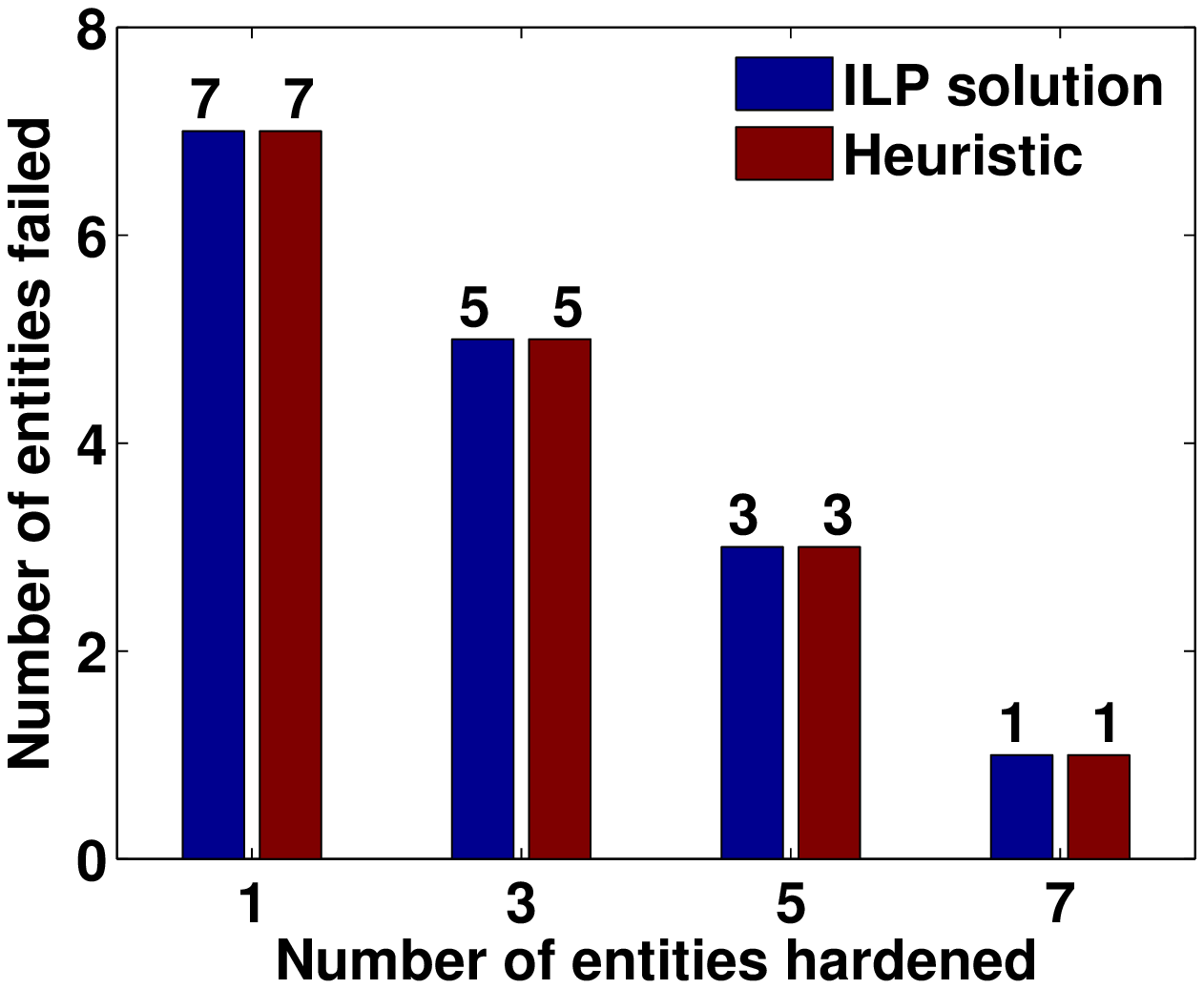}
		\caption{Region 5}\label{fig:plote}
	\end{subfigure}
	\caption{Comparison chart of the optimal solution (ILP) with the heuristic by varying number of entities hardened for each identified region}\label{fig:plot}
\end{figure*}

\section{Experimental Results}
\label{ExpRes}
In this section we present the experimental results of the Entity Hardening problem by comparing the optimal solution  computed using an ILP, and the proposed heuristic algorithm. The experiments were conducted on real world power grid data obtained from Platts (www.platts.com), and communication network data obtained from GeoTel (www.geo-tel.com) of Maricopa County, Arizona. The data consisted of $70$ power plants and $470$ transmission lines in the power network, and $2,690$ cell towers, $7,100$ fiber-lit buildings and $42,723$ fiber links in the communication network. We identified five non-intersecting geographical regions from the data set and labeled them from regions 1 through 5. For each of the regions, the entities of the power and communication network that were located within the geographic region formed the set $A$ and $B$ respectively. Each region was represented by an interdependent network $\mathcal{I}(A,B,\mathcal{F}(A,B))$. We use the IDR construction rules as defined in \cite{sen2014identification} to generate $\mathcal{F}(A,B)$.  

In all of our simulations IBM CPLEX Optimizer 12.5 to solve ILPs and Python 3 for heuristic is used. To analyze the Entity Hardening problem the value of $\mathcal{K}$ was set to $8$. The ILP in \cite{sen2014identification} was used to compute the $\mathcal{K}$ most vulnerable nodes in the network, and the set of failed entities due to the failure of the $\mathcal{K}$ entities was also computed. For the five regions, when the $\mathcal{K} = 8$ most vulnerable nodes failed, the total number of failed entities in the network were 28, 23, 28, 28 and 27 respectively. With the $\mathcal{K}$ most vulnerable nodes and final set of failed nodes as input, the ILP and heuristic of the Entity Hardening problem are compared with $k=1,3,5,7$. The results of these simulations are shown in Figure \ref{fig:plot}.
It is observed that the heuristic solution differs more from optimal at higher values of $k$ (factor of $0.5$ and $0.67$ for Regions 1 and 3 respectively with $k=7$). This is primarily because of the greedy nature of Algorithm 2. However on an average the heuristic solution differs by a factor of $0.13$ from the optimal.

\section{Conclusion}
\label{Conclusion}
In this paper we studied the entity hardening problem in multi-layer networks. We modeled the interdependencies shared between the networks using IIM, and formulated the the Entity Hardening problem in this setting. We showed that the problem is solvable in polynomial time for some cases, whereas for others it is NP-complete. We evaluated the efficacy of our heuristic using power and communication network data of Maricopa County, Arizona. Our experiments showed that our heuristic almost always produces near optimal results.

\begin{footnotesize}
\bibliographystyle{IEEEtran}
\bibliography{IEEEabrv,references,referencesBibToAdd}
\end{footnotesize}

\end{document}